\documentclass[copyright]{eptcs}
\providecommand{\event}{J.~Dassow, G.~Pighizzini, B.~Truthe (Eds.): 11th International Workshop\\
on Descriptional Complexity of Formal Systems (DCFS 2009)} 
\providecommand{\volume}{3}
\providecommand{\anno}{2009}
\providecommand{\firstpage}{85} 
\providecommand{\eid}{8}
\usepackage{breakurl}        

\input{dcfs.tex}
\usepackage{enumerate}

\usepackage{amsmath}
\usepackage{amssymb}
\usepackage{color}
\usepackage{graphicx}
\usepackage{subfigure}

\newcommand{\franzi}[1]{}

\newcommand{\ent}{\mbox{{\rm ent}}}
\newcommand{\ex}{\mbox{{\rm ex}}}
\newcommand{\Area}{\mbox{{\rm Area}}}

\newcommand{\states}{\mbox{{\rm states}}}
\renewcommand{\div}{\mbox{{\rm \ div\ }}}
\renewcommand{\mod}{\mbox{{\rm \ mod\ }}}

\newcommand{\alp}{\mbox{{\rm alph}}}

\newcommand{\N}{\mathbb{N}}

\newcommand{\suffix}{\leq_{{\rm s}}}

\newcommand{\oh}{\mathcal{O}}

\newcommand{\shuf}{\hspace{1mm}{\mathbin{\mathchoice
{\rule{.3pt}{1ex}\rule{.3em}{.3pt}\rule{.3pt}{1ex}
\rule{.3em}{.3pt}\rule{.3pt}{1ex}}
{\rule{.3pt}{1ex}\rule{.3em}{.3pt}\rule{.3pt}{1ex}
\rule{.3em}{.3pt}\rule{.3pt}{1ex}}
{\rule{.2pt}{.7ex}\rule{.2em}{.2pt}\rule{.2pt}{.7ex}
\rule{.2em}{.2pt}\rule{.2pt}{.7ex}}
{\rule{.3pt}{1ex}\rule{.3em}{.3pt}\rule{.3pt}{1ex}
\rule{.3em}{.3pt}\rule{.3pt}{1ex}}\mkern2mu}}\hspace{1mm}}

\newcommand{\shuft}[1]{\hspace{1mm}{\mathbin{\mathchoice
{\rule{.3pt}{1ex}\rule{.3em}{.3pt}\rule{.3pt}{1ex}
\rule{.3em}{.3pt}\rule{.3pt}{1ex}}
{\rule{.3pt}{1ex}\rule{.3em}{.3pt}\rule{.3pt}{1ex}
\rule{.3em}{.3pt}\rule{.3pt}{1ex}}
{\rule{.2pt}{.7ex}\rule{.2em}{.2pt}\rule{.2pt}{.7ex}
\rule{.2em}{.2pt}\rule{.2pt}{.7ex}}
{\rule{.3pt}{1ex}\rule{.3em}{.3pt}\rule{.3pt}{1ex}
\rule{.3em}{.3pt}\rule{.3pt}{1ex}}\mkern2mu}}_{#1}\hspace{1mm}}

\newcommand{\be}{\begin{equation}}
\newcommand{\ee}{\end{equation}}

\begin{document}

\title{On the Shuffle Automaton Size for Words\,%
\thanks{Research supported, in part, by the Natural Sciences and Engineering Research Council of Canada.}}
\def\titlerunning{On the Shuffle Automaton Size for Words}

\author{\makebox[0pt][c]{Franziska Biegler}\hspace{4cm}\makebox[0pt][c]{Mark Daley}
\institute{Department of Computer Science --
	University of Western Ontario\\
	London -- ON N6A 5B7 -- Canada}
\email{\makebox[0pt][c]{fbiegler@csd.uwo.ca}\hspace{4cm}\makebox[0pt][c]{daley@csd.uwo.ca}}
\and
Ian McQuillan
\institute{Department of Computer Science --
	University of Saskatchewan\\
	Saskatoon -- SK S7N 5A9 -- Canada}
\email{mcquillan@cs.usask.ca}
}
\def\authorrunning{F.~Biegler, M.~Daley, I.~McQuillan}
\maketitle

\begin{abstract}
We investigate the state size of DFAs accepting the shuffle of two words. We provide 
words $u$ and~$v$, 
such that the minimal DFA for $u\shuf v$ requires an exponential number of states. 
We also show some conditions for the words $u$ and $v$ which ensure a quadratic upper 
bound on the state size of~$u\shuf v$. Moreover, switching only two letters within 
one of $u$ or $v$ is enough to trigger the change from quadratic to exponential.
\end{abstract}

\section{Introduction}

Since its introduction, the shuffle operation 
has been aggressively studied as a model of nondeterministic interleaving in both purely theoretical and practical contexts. Perhaps due to the intrinsic nondeterminism of the operation, many problems concerning shuffle remain unsolved; e.g.,
shuffle decomposition for regular languages (though it is decidable~\cite{CSV01} for 
commutative regular languages or
locally testable languages  while for context-free languages
it is undecidable~\cite{CSV01}).

We follow here the recent trend of attacking the special case of the shuffle of two words, inspired by attempts to solve the decomposition problem. It has been shown in \cite{berstelwords} that shuffle decomposition
on individual words is unique as long as there are two letters used within the words. In \cite{shuffleTCS}, the result from~\cite{berstelwords} was extended to show that if two words $u$ and $v$ both contain at least two letters, then the
shuffle decomposition is the unique decomposition over arbitrary sets and not just words.

In this paper we ask a different type of question: what is the minimal state size for a DFA accepting the shuffle of two given words? For the more general case of languages, it has been shown in \cite{CSY} that the
shuffle of two DFAs can yield an
exponential minimal DFA ($\Omega(2^{nm}$), where $n,m$ were the sizes of the
two DFAs). We show here that DFAs accepting the shuffle of two words also require an exponential number of states in general; however, for words obeying certain conditions, a DFA may be constructed with, at most, quadratically many states.
 
A striking reminder of the complexity of the shuffle is operation is illustrated by showing that two words which may be accepted by a quadratically-bounded shuffle DFA can only be accepted by an exponentially large DFA when only two letters in one word are exchanged.


\section{Preliminaries}
%
Let $\mathbb{N}$ be the set of non-negative integers.
%
An alphabet $\Sigma$ is a finite, non-empty set of letters.  The set of all words 
over $\Sigma$ is denoted by $\Sigma^*$, and this set contains the empty word, $\lambda$. The set of all non-empty words over $\Sigma$ is denoted by $\Sigma^+$. 

Let $\Sigma$ be an alphabet and let $u,v\in \Sigma^*$. 
If $u=a_1^{\alpha_1} a_2^{\alpha_2}\cdots a_n^{\alpha_n}$ with $a_1,\ldots a_n\in\Sigma$, $\alpha_1,\ldots,\alpha_n\in\N$ and $a_i\neq a_{i+1}$, for $1\leq i<n$, then the {\em skeleton} of $u$ is defined as $\chi(u)=a_1a_2\cdots a_n$. The different occurences of the same letter $a$ in the skeleton of $u$ are called the {\em $a$-sections} of $u$. Furthermore, for $a\in\Sigma$, $|u|_a$ denotes the number of $a$'s in $u$. A word $u$ over $\Sigma$ is called non-repeating if $|u|_a\leq 1$ for all $a\in\Sigma$.
Let $u,v \in \Sigma^*$.  
The {\em shuffle} of $u$ and $v$ is defined as 
$$u\shuf v=\{u_1v_1\cdots u_nv_n\mid u=u_1\cdots u_n, v=v_1\cdots v_n, u_i \in \Sigma^*, v_i \in \Sigma^*, 1 \leq i \leq n\}.$$
We say
$u$ is a {\em suffix} of $v$, written $u \suffix v$, if $v = xu$, for some
$x \in \Sigma^*$.

A {\em trajectory} for two words $u$ and $v$ is a word $t\in\{0,1\}^*$, such that $|t|_0=|u|$ and $|t|_1=|v|$. Then the shuffle of $u$ and $v$ on $t$ is denoted by $u\shuft{t}v$ and is the unique string in $u\shuf v$, where a letter from $u$ is used whenever $t$ has a $0$ at the respective position, and a letter from $v$ is used whenever $t$ has a $1$. 
For details regarding shuffle on trajectories, consult \cite{mateescu98shuffle}.

We assume the reader to be familiar with nondeterministic and determinisitic finite automata. 
See \cite{HU,HB_Reg} for an introduction and more details on finite automata.
For each NFA we can effectively construct an equivalent DFA by using the so-called subset construction \cite{HU}. For an NFA with $n$ states, the DFA constructed this way can have up to $2^n$ states. 
There exists a unique minimal DFA (up to isomorphism) for each regular language. 
States $p$ and $q$ of a DFA are {\em distinguishable} if there exists $x$ such that 
$\delta(p,x)$ is a final state, but $\delta(q,x)$ is not, or vice versa.  Moreover, if every state of a DFA is accessible and every pair of states are distinguishable, then the DFA is minimal \cite{HU}.
For both NFAs and DFAs we use {\em size} synonynously with state size, and, thus, we define $|A|=|Q|$.

\section{Shuffle NFAs for words}
In this section we discuss basic properties of shuffle NFAs for two words.
\begin{definition}\label{naive_nfa}
Let $\Sigma$ be an alphabet and let $u=u_1\cdots u_m, v=v_1\cdots v_n\in \Sigma^+$, where $u_i,v_j\in\Sigma$ for all $1\leq i\leq m$ and $1\leq j\leq n$. We say $A$ is the naive shuffle NFA for $u$ and $v$ 
if 
$A=(Q,\Sigma,\delta,q_0,F)$ where 
$Q=\{0,\ldots,m\}\times\{0,\ldots,n\}$, $q_0=(m,n)$, $F=\{(0,0)\}$ and 

\vspace*{-4mm}

\begin{align*}
  \bullet&\text{ for $1\leq k\leq m$, $0\leq l\leq n$, we have 
$(k-1,l)\in \delta( (k,l), u_{(m-k+1)})$; and \textcolor{white}{}}\\ 
  \bullet&\text{ for $0\leq k\leq m$, $1\leq l\leq n$, we have 
  	$(k,l-1)\in \delta( (k,l), v_{(n-l+1)})$.}
\end{align*}


For all $i$ and $j$ with $1\leq i\leq m$ and $1\leq j\leq n$
we denote by $\overline{u}_i$ and $\overline{v}_j$ the suffixes of length $i$ and $j$ or the words $u$ and $v$, respectively.
We furthermore define $L_A(i,j)=\overline{u}_i\shuf \overline{v}_j$, which is accepted by the automaton $A'=(Q,\Sigma,\delta, (i,j),F)$. 
\end{definition}

Note that the automaton as defined above is not complete.
It is clear from Definition \ref{naive_nfa} 
that the naive shuffle NFA for $u$ and $v$ does in fact accept $u\shuf v$.

\begin{definition}\label{layers}
Let $A$ be the naive shuffle NFA for two words $u$ and $v$ over some alphabet $\Sigma$. 
The vertical layers and horizontal layers (shortly, v-layers and h-layers) are numbered
$0,1,\ldots, |u|+|v|$ and \hbox{$|u|,|u|-1,\ldots 1,0,-1,\ldots,-|v|$}, respectively. 
The {\em vertical layer} ({\em horizontal} respectively)  $k$, contains all states $(i,j)$ with $i+j=k$ (contains all states $(i,j)$ with $k=i-j$). 
\end{definition}

The vertical layer tells us how many letters we have read thus far, while the horizontal layer tells us the difference between the numbers of letters we have read from $u$ and $v$.
Note that the initial state $(|u|,|v|)$ is in horizontal layer $|u|-|v|$ if $|u|\geq|v|$, and in horizontal layer $|v|-|u|$ if $|v|\geq |u|$. 

\begin{definition}\label{areas}
Let $\Sigma$ be an alphabet and let $A$ be the naive shuffle NFA for some 
words $u,v\in\Sigma^+$. Let \hbox{$a\in\Sigma$} and $i_1,i_2,j_1,j_2\in\N$. 
Then $R=(a,(i_1,j_1),(i_2,j_2))$ is a nondeterministic area of $A$ if 
$$|u|\geq i_1\geq i_2\geq 0,\ |v|\geq j_1\geq j_2\geq 0$$
and
\begin{enumerate}[1.]
\item all states $(i,j)$ with $i_1\geq i> i_2$, $j_1\geq j> j_2$ are nondeterministic on $a$,
\item if they exist, $(i_1+1,j_1)$ and $(i_1,j_1+1)$  are determistic on $a$, and
\item $\delta((i_2,j_2),a)$ is undefined.
\end{enumerate} 
The set of all nondeterministic areas of $A$ is denoted by $\Area(A)$, 
and  we define the entrance and exit states of $R$ and the states in $R=(a,(i_1,j_1),(i_2,j_2))$ as 
\begin{align*}
\ent(R)&=\{(i_1,j)\mid j_1\geq j \geq j_2\}\cup \{(i,j_1)\mid i_1\geq i\geq i_2\};\\
\ex(R)&=\{(i_2,j)\mid j_1\geq j \geq j_2\}\cup \{(i,j_2)\mid i_1\geq i\geq i_2\};\\
\states(R)&=\{(i,j)\mid i_1\geq i>i_2, j_1\geq j>j_2\}.
\end{align*}
\end{definition}



\begin{example}
Let $u=bbaa$, $v=aab$. Then the naive shuffle NFA $A$ for $u$ and $v$ 
has $$\Area(A)=\{(a,(2,3),(0,1)),(b,(4,1),(2,0))\}.$$ 
$A$ is depicted twice in 
Figure~\ref{Fig:Layer_Area}, first with the different horizontal and vertical 
layers labelled and then with the nondeterminisitic areas shown in grey.
\begin{figure}[!htbp]
\centering
\subfigure{
\includegraphics[scale=.4]{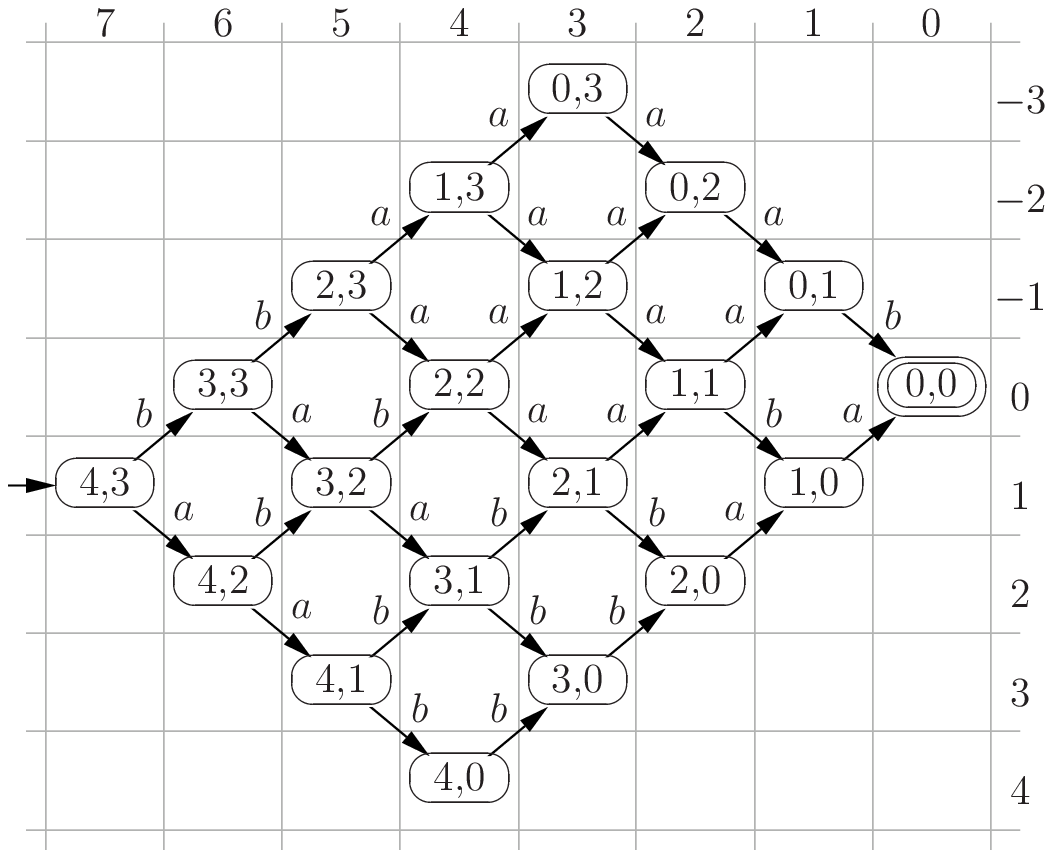}
}
\subfigure{
\includegraphics[scale=.4]{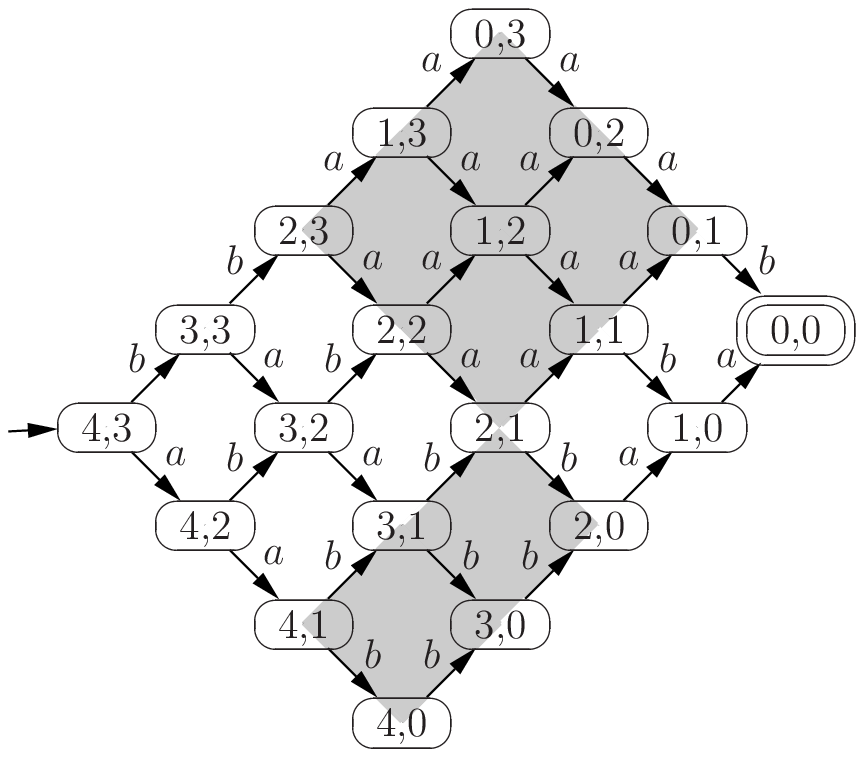}
}
\label{Fig:Layer_Area}
\caption[]{Naive shuffle NFA for $u=bbaa$ and $v=aab$}
\end{figure}
\vspace{-9mm}

\hspace*{\fill}$\diamond$
\end{example}

We know from \cite{finite_nfa2dfa} that given an NFA accepting a finite language over a $k$ letter alphabet with $q$ states, a minimal DFA accepting the same language has at most $\oh(k^{\frac{q}{\log_2(k)+1}})$ states in the worst case. Thus for a binary alphabet, $\oh(2^{\sqrt{q}})$ states are both necessary and sufficient in the worst case.

In the case of naive shuffle NFAs, it is immediately obvious that during a subset-construction only state labels from the same vertical layer can appear within the same state of the DFA. 
If $|u|=m$ and $|v|=n$ with $0\leq n\leq m$, then for each number between 1 and $n$ there are two vertical layers with that number of states, and there are $(m-n+1)$ vertical layers with $(n+1)$ states.
If we assume that for each v-layer, all subsets of states except the empty set are possible (it is sufficient to add the empty set once)
then this gives us an upper bound of 
\begin{equation}\label{naive_bound}
2\sum_{i=1}^n (2^i-1) + (m-n+1)(2^{n+1}-1)+1= 2^{n+1}(m-n+3)-m-n-4\pagebreak[4]
\end{equation}
for the number of states in the equivalent DFA. Recall that the NFA has $(m+1)(n+1)$ states, so 
the bound in (\ref{naive_bound}) is better than the bound $\oh(k^\frac{(m+1)(n+1)}{\log_2(k)+1})$ (where $k$ is the size of the alphabet) from \cite{finite_nfa2dfa} for arbitrary finite languages.

When $u$ and $v$ are over disjoint alphabets then the naive shuffle NFA for $u$ and $v$ is also the minimal $DFA$ for $u\shuf v$.  This can be seen as every pair of states that are not distinguishable would have to be in the same vertical layer, however, every two states in the same layer have some different path to the final state. Thus, all pairs of states are distinguishable.  So, in the worst case there is a lower bound of $(|u|+1)\cdot (|v|+1)$ on the size of the shuffle DFA for $u$ and $v$.

We can also see that the bound (\ref{naive_bound}) is not tight, as only labels of states of the NFA which have identical Parikh vectors can appear together as the label of a state in the DFA. 
Thus the bound (\ref{naive_bound}) would be reached only if 
$u,v\in\{a\}^*$ for some $a\in\Sigma$. But then the minimal DFA for $u\shuf v$ would only have $|u|+|v|+1$ states, a contradiction.


\begin{definition}
Let $u$ and $v$ be words over some finite alphabet $\Sigma$ and let $A$ be the naive shuffle NFA for $u$ and $v$. 
A walk through $A$ is a sequence of states $s_0, s_1, \ldots, s_{|u|+|v|}$, where $s_0=(|u|,|v|)$, $s_{|u|+|v|}=(0,0)$, and for all $i$ with $0\leq i <|u|+|v|$, we have $s_{i+1}\in\delta(s_i,a )$ for some $a \in \Sigma$. 
We say that a given vertical or horizontal layer is visited $x$-times during a given walk if exactly $x$ states from that layer appear in the walk. 
\end{definition}

Note that there exists a bijective mapping between the walks through a naive shuffle NFA and the set of possible trajectories for the shuffle of $u$ and $v$. 

\begin{lemma}\label{must_visit_layers}
Let $u,v$ be words over some alphabet $\Sigma$ and let $A$ be the naive shuffle NFA for $u$ and $v$. Then during each walk through $A$, every vertical layer has to be visited exactly once, while each horizontal layer may be visited once, multiple times or not at all. 
However, if $|u|\geq |v|$ then each of the horizontal layers $0,1,\ldots,|u|-|v|$ has to be visited at least once, and similarly
if $|v|\geq |u|$.
\end{lemma}

\section{Shuffle DFAs for periodic words}

In this section we focus on a special case of the shuffle of two words, namely the shuffle of two words that are periods of a common underlying word. 
Thus 
$u=w_1w^k$ and $v=w_2w^l$, where $w\in\Sigma^+$, $w\notin a^+$ for any $a\in\Sigma$,  $k,l\geq 0$ and both $w_1$ and $w_2$ are suffixes of $w$. 
At first glance one could assume that these words lead to an exponential blow-up in the state size when converting the naive shuffle NFA to a DFA, because they induce long common factors. However we will show that this is not the case when the underlying word $w$ contains at most one section per letter in $\Sigma$.
%
We first show two subset-relations between different periodic shuffles over the same underlying word. These subset-relations are then used to construct the DFA in a more efficient manner.


\begin{lemma}\label{period_change_subset}
Let $\Sigma$ be a finite alphabet and let $w=a_1\cdots a_n$ for some $n\geq 2$, such that $\alp(w)\geq 2$. 
Let $u=w_1w^k$, $v=w_2w^l$, $u'=w_1w^{k'}$, $v'=w_2w^{l'}$ where $0\leq l< k' <k$, $0\leq l< l' <k$, $k+l=k'+l'$ and $w_1,w_2$ are both either empty or proper suffixes of $w$.
Then $u\shuf v \subsetneq u'\shuf v'$.
\end{lemma}

\begin{proof}
Let $A$ be the naive shuffle NFA for $u$ and $v$.
Let $t$ be a trajectory for $u$ and $v$. We construct a trajectory $t'$ for $u'$ and $v'$, such that 
$u\shuft{t} v = u'\shuft{t'}v'$. 

As discussed in Lemma \ref{must_visit_layers}, the horizontal layers $0,\ldots, |u|-|v|$ have to be visited at least once during any walk through $A$.
Let $p\leq |u|-|v|$ be maximal such that $p\mod n=0$. 
Thus $p\geq |u|-|v|-n$, which implies, as $|u|-|v|\geq n$, that layer
$p$ has to be visited at least once during any walk through $A$.
Let $p'=p-n(l'-l)$ (see Figure \ref{figure2}). Then
$$p'\geq |u|-|v|-n-l'n+ln = kn-ln+|w_1|-|w_2|-n-l'n+ln > kn-l'n-2n\geq -n.$$
Thus $p'>-n$, but as $p'\mod n=0$, this implies that $p'\geq 0$ and, thus,
$p'$ is also visited at least once during any walk through $A$.


We let $(i,j)$ be the first occurence of a state in h-layer $p$ in $u\shuft{t}v$
and we let $(i',j')$ be the first occurrence of a state in h-layer $p'$ in $u\shuft{t}v$. 

Then $i\mod n = j\mod n$ and $i'\mod n = j'\mod n$, which means that when in states $(i,j)$ and $(i',j')$ we are at the same point in the underlying period $w$ for both words $u$ and $v$. 
Let $t=t_1t_2t_3$ where $t_1$ is the part of $t$ before visiting $(i,j)$, $t_2$ is the part of $t$ after visiting $(i,j)$ but before visiting $(i',j')$ and $t_3$ is the part of $t$ after visiting $(i',j')$. Then $|t_2|_1=|t_2|_0+n(l'-l)$.
%
Now let $t'=t_1\overline{t}_2t_3$, where $\overline{t}_2$ is obtained from $t_2$ by switching all $0$'s for $1$'s and vice versa. 
Then $u'\shuft{t'} v'= u\shuft{t} v$.

\begin{figure}[!htbp]
\centering
\includegraphics[scale=.6]{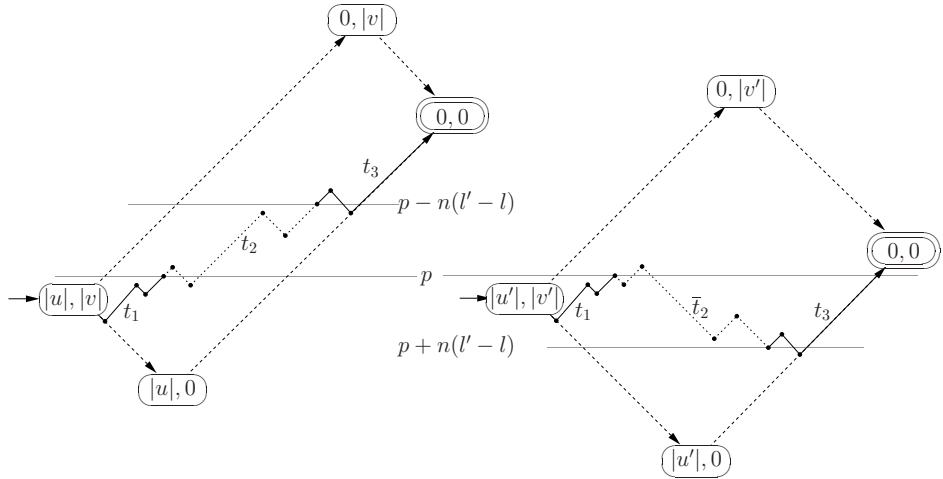}
\caption[]{Transformation of a trajectory  by switching all $0$'s and $1$'s in $t_2$. \label{Fig:traj_trans}}
\label{figure2}
\end{figure}

In order to show that the inclusion is proper, assume without loss of generality that $k'\geq l'$ and $|w_1|= |w_2|+q$, $q\geq 0$. 
We define an infinite word $\overline{w}=a_1^2\cdots a_n^2a_1^2\cdots a_n^2\cdots$ and\nolinebreak\ 
let $z= u'\shuft{t_z} v'$ where $t_z=0^{q}(01)^{|w_2|+l'n}0^{(k'-l')n}$. Then $z=z_1z_2z_3$, where $|z_1|=q$, $|z_2|=|w_2|+l'n$. 
Then $z_2$ is a factor of $\overline{w}$. But the length of factors of $\overline{w}$ that are also factors of a word in $u\shuf v$ is bounded by 
$|w_2|+ln<|w_2|+l'n$, hence $z\in u'\shuf v'\setminus u\shuf v$.
\end{proof}

The next result is similar to the previous one, but now only the suffixes of $w$ at the beginning of the words are swapped and the number of repetitions of $w$ do not change. The proof is omitted due to space.

\begin{lemma}\label{period_nonchange_subset}
Let $\Sigma$ be a finite alphabet and let $w=a_1\cdots a_n$ for some $n\geq 2$, such that $\alp(w)\geq 2$. 
Let $u=w_1w^k$, $v=w_2w^l$, $u'=w_2w^k$, $v'=w_1w^l$ where $0\leq l< k$ and $w_2<_s w_1\leq_s w$. 
Then $u\shuf v \subsetneq u'\shuf v'$.
\end{lemma}

We can use Lemma \ref{period_change_subset} and Lemma \ref{period_nonchange_subset} to show a subset-relation between the languages defined by certain states of the naive shuffle NFA for two words that are periodic over the same underlying word. This result will be useful in the next subsection to show that the minimal DFA for the shuffle of periodic words over certain underlying words is smaller than the naive NFA for these words.

\begin{lemma}\label{state_period_subsets}
Let $u=w_1w^k$ and $v=w_2w^l$, where $w=a_1\cdots a_n$ for some $n\geq 1$ 
such that $ a_1,\ldots a_n\in\Sigma$ and 
$w_1$ and $w_2$ are suffixes of $w$. 
Let $A$ be the naive shuffle NFA for $u$ and $v$ and let $i,j,i',j'$ be natural numbers such that
\begin{enumerate}[{\rm 1.}]
\item $1\leq i\leq |u|$, $1\leq i'\leq |u|$, $1\leq j\leq |v|$, $1\leq j'\leq |v|$;
\item  $i+j=i'+j'$; 
\item  $\{i\mod n, j\mod n\} = \{i'\mod n,j'\mod n\}$; and
\item  $|i-j|\geq|i'-j'|$. 
\end{enumerate}
Then $L_A(i,j)\subseteq L_A(i',j')$, and $L_A(i,j)= L_A(i',j')$ if and only if $\{i,j\}=\{i',j'\}$.
\end{lemma}

\begin{proof}
Obviously $\{i,j\}=\{i',j'\}$ implies $L_A(i,j)= L_A(i',j')$, so we only have to show that Conditions~1, 2, 3 and $|i-j|>|i'-j'|$ imply that $L_A(i,j)\subsetneq L_A(i',j')$. 

By Condition 1 there exist suffixes $\overline{u}_i$ and $\overline{u}_{i'}$ of $u$ and suffixes $\overline{v}_j$ and $\overline{v}_{j'}$ of $v$, such that 
$$L(i,j)=\overline{u}_i\shuf \overline{v}_j \mbox{ and } 
L(i',j')=\overline{u}_{i'}\shuf \overline{v}_{j'}.$$
Condition 3 implies that there exist suffixes $\overline{w}_1,\overline{w}_2$ of $w$, such that $\{\overline{u}_i,\overline{v}_j\}=\{\overline{w}_1w^p,\overline{w}_2w^q\}$ and 
$\{\overline{u}_{i'},\overline{v}_{j'}\}=\{\overline{w}_1w^{p'},\overline{w}_2w^{q'}\}$ for some $p,q,p',q'\geq 0$ and $p,p'\leq k$ and $q,q'\leq l$.
Furthermore, Condition~2 implies that the words in $L(i,j)$ and $L(i',j')$ all have the same length, which implies $p+q=p'+q'$. 
Now we get two cases, depending on whether $i$ and $j$ are from the same iteration of $w$ as $i'$ and $j'$ or not.

If $\{i\div n,j\div n\}=\{i'\div n,j'\div n\}$, then $\{p,q\}=\{p',q'\}$. 
Then $|i-j|>|i'-j'|$ implies that $p=q'$ and $q=p'$ and either both $|\overline{w}_1|>|\overline{w}_2|$ and $p>q$, or both $|\overline{w}_1|<|\overline{w}_2|$ and $p<q$. We assume the former without loss of generality and obtain $L(i,j)\subsetneq L(i',j')$ by Lemma~\ref{period_nonchange_subset}.

If $\{i\div n,j\div n\}\neq \{i'\div n,j'\div n\}$ then $\{i,j\}\neq\{i',j'\}$ follows immediately. 
Thus, by Condition 4, we have $|i-j|>|i'-j'|$, which implies without loss of generality that $q<q'<p$ and $q<p'<p$ (the case where $p<q'<q$ and $p<p'<q$ is symmetric). But this implies that $L(i,j)\subsetneq L(i',j')$ by Lemma~\ref{period_change_subset}.%
%
\end{proof}

\franzi{\input{third_subset_relation}}

\subsection{Underlying non-repeating words}

We now show that the shuffle of periodic words over a non-repeating $w$ yields deterministic finite automata that have at most a quadratic number of states. 

\franzi{input{naive_subset_construction} here}


\begin{theorem}\label{period_dfa_greater}
Let $u=w_1w^k$ and $v=w_2w^l$, where $k>l\geq 0$ and $w=a_1\cdots a_n$ for some $n\geq 2$ such that $a_i=a_j$ implies $i=j$ whenever $1\leq i,j\leq n$ and $w_1$, $w_2$ are non-empty suffixes of $w$. 
If $k>l$, then the minimal DFA for $u\shuf v$ has 
$(|u|+1)\cdot (|v|+1)-\frac{1}{2}(|v|)\cdot(|v|+1)-\frac{1}{2}m\cdot(m+1)$ states, where $m\leq |v|$ is maximal such that $(|u|-m)\mod n=0$.
If $k=l$, $|w_1|\geq |w_2|$, then the minimal DFA for $u\shuf u$ has $(|u|+1)(|v|+1)-\frac{1}{2}|v|\cdot (|v|+1)-\frac{1}{2}(|u|-|w|)(|u|-|w|+1)$ states.
\end{theorem}

\begin{proof}
Assume first that $k>l$.
We construct the naive shuffle NFA $$A=(Q,\Sigma,\delta,s_0,F)$$ for $u$ and $v$. 
Obviously $|Q|=(|u|+1)\cdot(|v|+1)$ by Definition \ref{naive_nfa}.
In the following we perform several transformations with the automaton $A$, so that in the end $A$ has the properties that are mentioned in the theorem statement.

\noindent
{\bf Removing {\boldmath $\frac{1}{2}|v|\cdot(|v|+1)+\frac{1}{2}m\cdot(m+1)$} states:}
We look at the horizontal layer~$0$, which contains the final state $(0,0)$ as well as the states $(|v|,|v|),\ldots,(1,1)$. All the states in this layer, except the final state are nondeterminisitic, so for all $i$ with $0<i\leq |v|$ there exists $a\in\Sigma$, such that 
$$\delta((i,i),a)=\{(i-1,i),(i,i-1)\}.$$
By Lemma \ref{state_period_subsets} we know that $L(i,i-1)=L(i-1,i)$. Thus, we can modify the transition function $\delta$ to $\delta((i,i),a)=(i,i-1)$ without changing the accepted language. When we have done this 
for all nondeterministic states in the horizontal layer $0$, the states in the horizontal layers $-1,\ldots,-|v|$ are unreachable 
and can be removed from $Q$. The number of states removed in this way is $\sum_{i=1}^{|v|} i= \frac{1}{2} |v|\cdot(|v|+1)$.

We now look at the horizontal layer $|u|-m$, which contains the states 
$$(|u|,m),(|u|-1,m-1),\ldots, (|u|-m,0).$$ 
As $m\leq |v|$ is maximal, such that \hbox{$(|u|-m)\mod n=0$}, the horizontal layers 
$$(|u|-|v|), (|u|-(|v|-1)), \ldots, (|u|-(m+1))$$
do not contain any nondeterministic states.

Furthermore, we know that all states in the horizontal layer $|u|-m$ except for the state $(|u|-m,0)$ are nondeterminisitic. 
Thus, if we let $(i,j)$ be one of the nondeterministic states in the horizontal layer 
$|u|-m$, then $(i,j)=(|u|-p,m-p)$ for some $0\leq p<m$ and there exists $a\in\Sigma$, such 
that $$\delta( (i,j),a)=\{(i-1,j),(i,j-1)\}.$$ 
This implies that $\{(i-1)\mod n, j\mod n\}=\{i\mod n, (j-1)\mod n\}$, as the outgoing transitions of both states $(i-1,j)$ and $(i,j-1)$ carry the same labels and $w$ is non-repeating. 
Also it is obvious that $(i-1)+j= i+(j-1)$ and 
$1\leq i\leq |u|$, $1\leq i-1\leq |u|$, $1\leq j\leq |v|$, $1\leq j-1\leq |v|$. Furthermore as $|u|>|v|$, we have $|(i-1)-j|=||u|-|v|-1|<||u|-|v|+1|=|i-(j-1)|$.
Therefore by Lemma~\ref{state_period_subsets} we have 
$L(i-1,j)\subsetneq L(i,j-1)$, which implies that we can modify the transition function $\delta$ of $A$ to $\delta((i,j),a)=(i,j-1)$ without changing the accepted language. 
Once we have done that for all states in the horizontal layer $|u|-m$, the states in horizontal layers $|u|-m+1,\ldots,|u|$ are no longer reachable and can be removed.
The number of states removed in this way is $\sum_{i=1}^{m} i= \frac{1}{2} m\cdot(m+1)$.

We now have $|Q|=(|u|+1)\cdot (|v|+1)-\frac{1}{2}(|v|)\cdot(|v|+1)-\frac{1}{2}m\cdot(m+1)$, as claimed in the Theorem statement, however 
$A$ could still be nondeterministic.

\noindent
{\bf Removing remaining nondeterminism:}
The only horizontal layers left in $A'$ are $|u|-m,\ldots,0$. Furthermore we have already removed all nondeterminism from the horizontal layers $|u|-m$ and $0$. Also note that all states $(i,j)\in Q$ now have $i\geq j$ and the only states with $i=j$ are those in the horizontal layer $0$.
Thus, all remaining nondetermism must occur in the horizontal layers $|u|-m-1,\ldots,1$.
However, a state $(i,j)$ is nondeterministic precisely when $i\mod n= j\mod n$, which is only possible for states in the horizontal layers $|u|-m-pn$ where $1\leq p<k-l$. 
Let $(i,j)$ be such a state. 
 As $(i,j)$ has precisely two outgoing transitions, this implies that there exists a letter $a\in\Sigma$, such that 
 $\delta( (i,j),a)=\{(i-1,j),(i,j-1)\}$.
As no letter appears more than once in $w$, we know that 
$$\{(i-1)\mod n, j\mod n\}=\{i\mod n, (j-1)\mod n\}.$$ 
Thus, as $i>j$ implies that $|(i-1)-j|<|i-(j-1)|$, which implies, 
by Lemma \ref{state_period_subsets},  
$$L(i,j-1)\subsetneq L(i-1,j).$$
We can, thus, redefine $\delta((i,j),a)=(i-1,j)$. 

\noindent
{\bf Showing minimality:}
We can show that $A$ is minimal by induction on the layers. The details of this part of the proof are omitted due to space.

If $k=l$, then there are fewer than $|w|$ horizontal layers between the initial and final state and the proof has to be changed slightly. The proof for this case in omitted due to space.
\end{proof}

\begin{example}
Let $u=bc(abc)^2$, $v=abc$. Then the naive NFA for $u\shuf v$ is shown on the left side of Figure \ref{Fig::abc3_abc_cons}.
According to the proof of Lemma \ref{period_dfa_greater}, we can remove all the shaded states and transitions and we can furthermore also remove the dashed non-shaded transitions. This then leaves the minimal DFA for $u\shuf v$, as shown on the right side of Figure \ref{Fig::abc3_abc_cons}.
\begin{figure}[!htbp]
\begin{center}
\includegraphics[scale=.6]{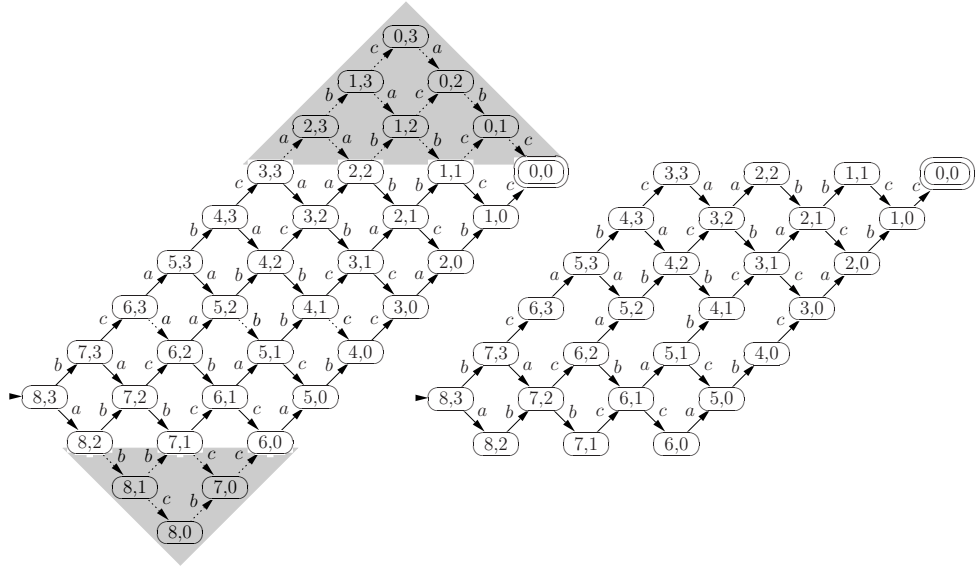}
\end{center}
\caption{Naive shuffle NFA  and minimal shuffle DFA for $u=bc(abc)^2$ and $v=abc$.\label{Fig::abc3_abc_cons}}
\end{figure}
\hspace*{\fill}$\diamond$
\end{example}

From the proof of Theorem \ref{period_dfa_greater} it is immediate that we can construct the minimal shuffle DFA for periodic words over a non-repeating underlying word directly without first constructing the NFA.

\subsection{Periodic words with one section per letter}

We now generalize Theorem \ref{period_dfa_greater} to underlying words the skeletons of which are non-repeating.
That is, we still consider only words $u=w_1w^k$ and $v=w_2w^l$, where $k\geq l\geq 0$ and $w_1$ and $w_2$ are proper (possibly empty) suffixes of $w$. However, $w$ no longer has to be non-repeating, but we now have $w=a_1^{p_1}\cdots a_n^{p_n}$ for some $n\geq 2$ and positive integers $p_1,\ldots, p_n$ and where $a_1\cdots a_n$ is non-repeating.
 

\begin{theorem}\label{nonrepskel}
Let $\Sigma$ be a finite alphabet and let $w\in\Sigma^+$, such that $|w|=n\geq 2$ and for all $a\in\Sigma$, we have $|\chi(w)|_a\leq 1$.  
Let $u=w_1w^k$ and $v=w_2w^l$ where $w_1,w_2$ are suffixes of $w$ and $k,l\geq 0$. 
Then there exists a DFA $A$ with $L(A)= u\shuf v$ and $|A|\in\oh(|u|\cdot |v|)$.
\end{theorem}

\begin{proof}
Let $A'=(Q',\Sigma,\delta',Q_0',F')$ be the naive shuffle NFA for $u$ and $v$.
Obviously 
$$|A'|=(|u|+1)(|v|+1).$$
We show that for each nondeterministic area $R\in\Area(A')$, we can determinize $R$ in such a way, by using Lemma \ref{state_period_subsets}, that no state in the DFA contains more than one label from $\ex(R)$, and no more than $\oh(|\states(R)\cup\ex(R)|)$ contain labels from $\states(R)\cup\ex(R)$.

Let $R=(a,(i_1,j_1),(i_2,j_2))\in\Area(A')$ and let $(i,j)\in\ent(R)$. 
When determinizing $R$ by using a subset construction it is easy to see that if both states $(i',j')\in Q'$ and $(i'',j'')\in Q'$ can be reached from $(i,j)$ by reading~$k$ $a$'s for some $k\in\N$, then also all states $(\overline{i},\overline{j})$ with $\overline{i}+\overline{j}=i'+j'$ and either both 
$i'\leq \overline{i} \leq i''$ and $j'\geq\overline{j}\geq j''$ or both $i'\geq \overline{i} \geq i''$ and $j'\leq\overline{j}\leq j''$ can be reached from $(i,j)$ by reading~$k$ $a$'s. Furthermore if some state $(i',j')$ can be reached from $(i,j)$ by reading $k$ $a$'s, then also some state $(i''',j''')\in\ent(R)\cup\ex(R)$ can be reached from $(i,j)$ by reading $k$ $a$'s.
This implies that at most $2|\states(R)\cup\ex(R)|$ states can result from a subset construction on $R$, assuming that we are starting with states that contain only individual entrance state labels.

It is also obvious that each state $q$ obtained by performing a subset construction on $R$ contains at most $2$ exit state labels (as there are only two exit states of~$R$ per vertical layer). 
If there is at most one exit state of $R$ in $q$, then $q$ does not induce any states with multiple labels outside of the states in $\states(R)\cup\ex(R)$ and we are done. 
If $q$ contains distinct exit states $(i',j')$ and $(i'',j'')$ then there exists an $n\in\N$, with $1\leq n\leq(i_1-i_2)$ such that 
$(i',j')=(i_2-1,j_2+n)$ and $(i'',j'')=(i_2+n,j_2-1)$ (or vice versa). 
But then, as $i_2\mod n=j_2\mod n$, we know that $\{i'\mod n,j'\mod n\}=\{i''\mod n,j''\mod n\}$. Furthermore we know that $i'+j'=i''+j''$ and either $|i'-j'|\geq |i''-j''|$ or $|i'-j'|<|i''-j''|$. 
Thus by Lemma \ref{state_period_subsets} we have either $L_A(i',j')\subseteq L_A(i'',j'')$ (if $|i'-j'|\geq|i''-j''|$) 
or $L_{A'}(i'',j'')\subset L_{A'}(i',j')$ (if $|i''-j''|>|i'-j'|$) and, hence, we can remove one of $(i',j')$ and $(i'',j'')$ from $q$ without changing the accepted language. 

Thus, the nondeterministic areas do not induce any states with multiple layers outside of the 
nondeterminisitic areas, which implies that $|A|\in\oh(|u|\cdot|v|)$. 
\end{proof}

\section{Exponential shuffle automata}

\begin{theorem}\label{1point09}
Let $\Sigma$ be an alphabet of size at least $2$.  Then there exist words $u,v\in\Sigma^+$, $|u|=|v|$, such that the size of the minimal DFA accepting $u \shuf v$,  is $\Omega(\sqrt[8]{2}^{|u|} )$.
\end{theorem}

Note that, in the proof below, the numbering of the layers is different from the numbering used thus far. 

\begin{proof}
For $n>1$, let 

\vspace*{-4mm}

\begin{align*}
u_n &= (aabb)^n aabbaabb (aabb)^n aaaaa, v_n = (aabb)^n aabababb (aabb)^n bbbbb,\\ 
X_n &= a(aabb)^n aaa(bbbbaaaa+bbbabaaa)^{n+1}bbbb(aabb)^n aaaaa bbbbb.
\end{align*}


Let $A_n=(Q, \Sigma, q_0, F, \delta)$ 
be the naive shuffle NFA for $u_n$ and $v_n$.  We have $A_2$ pictured in Figure \ref{mainfigure}.

\begin{figure}[t]
\begin{center}
\includegraphics[scale=.55]{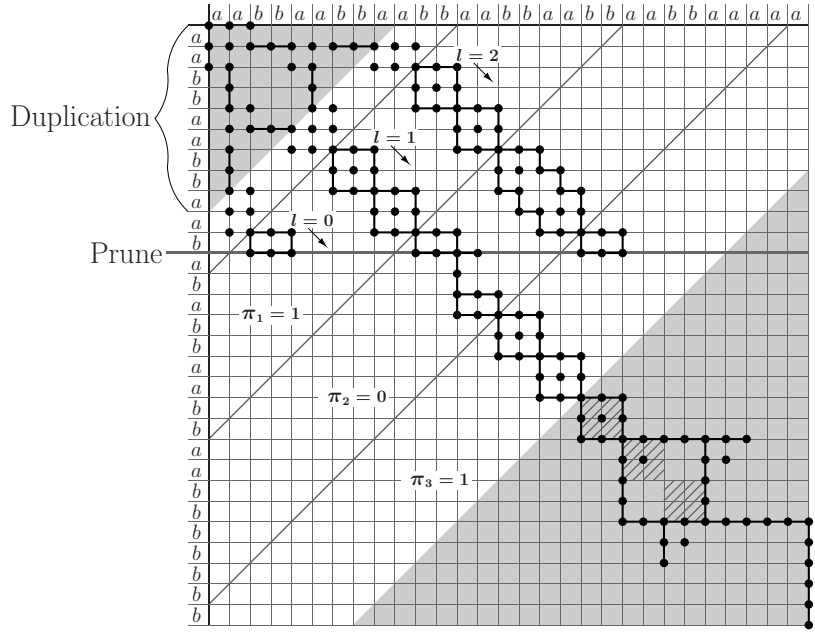}
\end{center}
\caption{The diagram is the naive NFA $A_2$, with the top left corner as the initial state, the bottom right corner being the final state, and the lines of the grid being transitions on the letter labelling the axis, with $u_2$ along the horizontal and $v_2$ along the vertical axis.  The input to $A_2$ is $a(aabb)^2aaa(bbbbaaaa)(bbbabaaa)(bbbbaaaa)bbbb(aabb)^2 aaaaabbbbb,$ with active states marked with bullet points.}
\label{mainfigure}
\end{figure}

Let $m = |v_n| = |u_n| = 8n+13$,
and there are $2(8n+13)+1 = 16n+27$ vertical layers.  For each layer $i$, let $Q_i$ be the set of states in that layer.  Let $q_{i,j}$ be the $j$th state (along the diagonal) in the $i$th layer.  There are $i$ states in the $i$th layer for $i \leq 8n+14$ and $(8n+14) - (i-(8n+14)) = 16n+28-i$ for $8n+14<i$.  
For each $w$ which is a prefix of some word in $u_n\shuf v_n$, let $Q_w$ be the set of states $\delta(q_0, w)$.  We will only consider input 
words in $X_n$.  
In Figure \ref{mainfigure}, we have the set of states $Q_w$, with each state denoted by bullet points.

We show by induction that for each $i$, $1 \leq i \leq n, Q_{a(aabb)^i} = 
\{q_{4i+2, j}, q_{4i+2, j+3} \mid j = 2 +4l, 0 \leq l < i\}$.  This is the ``duplication stage'', consisting of the states 
in the shaded top left corner of Figure \ref{mainfigure}. The details of this part of the proof are omitted due to space.


Thus, after reading $a(aabb)^n$, we are in one of the states in 
$$Q_{a(aabb)^n} = \{ q_{4n+2,j}, q_{4n+2, j+3} \mid j = 2+4l, 0 \leq l <n\}.$$
This occurs at the bottom diagonal of the ``duplication'' section in Figure \ref{mainfigure}.
Then $$Q_{a(aabb)^n aaa} = \{ q_{4n+5,j} \mid j = 3+4l, 0 \leq l \leq n \}$$
which is of size $n+1$.  The next set of input letters is in 
$(bbbbaaaa + bbbabaaa)^{n+1}$.  This is the so called ``filtering stage'', marked in white in Figure \ref{mainfigure}.
Intuitively, each element of $Q_{a(aabb)^n aaa}$, as determined by $l$,  will continue roughly along
a diagonal (we get a diagonal for $l$ being $0,1,2$ in the figure) until each reaches $baba$ along $v$ marked by the ``prune'' line of the figure.  If the input is then $bbbbaaaa$, this diagonal gets
``cut off'', while all other states in the vertical layer are able to continue along its diagonal.
However, if the input is $bbbabaaa$, then every diagonal in the vertical layer is able to continue.  
Since each diagonal reaches the ``prune'' line at a different time, 
we can selectively keep or remove each diagonal one at a time.

More formally, assume that $x_1 \cdots x_{n+1}$ is the input, $x_i \hspace*{-.7mm} \in (bbbbaaaa+bbbabaaa)$.
Let $\pi_i = 0$ if $x_i = bbbabaaa$, and $\pi_i= 1$ if $x_i =bbbbaaaa$. 
The sections of $A_2$ when reading $x_1,x_2, x_3$ are separated by lines in Figure \ref{mainfigure} where $\pi_1=1, \pi_2 = 0, \pi_3 = 1$.
We can then show by induction that for each $i$, $1 \leq i \leq n+1$,
$Q_{a(aabb)^n aaax_1 \cdots x_i} = \{q_{4n+5+8i,j} \mid j = 3+ 4l+4i, 0 \leq l \leq n, 
(l<i \Rightarrow \pi_l = 1)\}.$

The details of this part of the proof are omitted due to space.


Hence, $Q_{a(aabb)^n aaax_1 \cdots x_{n+1}} = \{q_{4n+5 + 8(n+1),j}\mid
j = 3+4l +4(n+1), \pi_l = 0\}$.
No matter the contents of this set, which depends on $x_1, \cdots, x_{n+1}$, every
state can reach a final state on $bbbb(aabb)^n aaaaa bbbbb$ since the rest of $u$
is of the form $bb(aabb)^*aaaaa$ and the rest of $v$ is of the form $bb(aabb)^*bbbbb$.
Therefore, if we use the subset construction \cite{HU} on $A_n$, there is only
one set of states we can be in after reading each prefix of $a(aabb)^naaa$.  
As we read each prefix $w$ of
$x_1 \cdots x_{n+1}, w = x_1 \cdots x_i y, |y|<8, x_j \in (bbbbaaaa+bbbabaaa), j \leq i$, then
$q_{4n + 5 + 8i, 3 + 4l + 4i} \in Q_{a(aabb)^n x_1 \cdots x_i}$ if and only if
$l \geq i$ or $\pi_l = 0$.  There are $2^i$ such subsets.  And indeed, if $|y| \geq 4$,
then $\delta(q_{4n+5+8i, 3 + 4l + 4i},y)$ is undefined if and only if 
$\pi_{i+1} = 1$.  Hence, after reading each  prefix of length $1$ to
$|x_1 \cdots x_{n+1}|$, there are

\vspace*{-4mm}

\begin{eqnarray*}
&& 3+8 \cdot 2^1 + 8 \cdot 2^2 + \cdots + 8\cdot 2^n + 5\cdot 2^{n+1}
= 3 + 5 \cdot 2^{n+1} + 8(2^1 + \cdots + 2^n)\\
& = &  3+5 \cdot 2^{n+1} + 8(2^{n+1}-2) = 13(2^{n+1}) -13 =13(2^{n+1}-1)
\end{eqnarray*}

\vspace*{-1mm}

\noindent
sets of states created in the subset construction.  Thus,
when reading every prefix of $$a(aabb)^naaa x_1 \cdots x_{n+1},$$ 
$4(n+1) + 13(2^{n+1}-1)$ sets of states are created and thus the subset construction
requires at least this many states, and the remaining
input is of length $4(n+1)+10$, the automaton from the subset construction has at least
$8(n+1) + 13(2^{n+1}-1) + 10$ states.  

We can now show that the minimal automaton created from this subset construction automaton requires this many states as well, by showing that there are at least this many distinguishable states~\cite{HU}.  
The details of this part of the proof are omitted due to space.
Hence, we get $\Omega(\sqrt[8]{2}^{m} )$ where $|u_n|=|v_n|=m$.
\end{proof}

Theorem \ref{1point09} is especially interesting in light of Theorem \ref{nonrepskel}, which showed that the minimal DFA for the shuffle of $u=(aabb)^{2n+2}$ and $v=(aabb)^{2n+2}$ is in $\oh(n^2)$. It is easy to see that adding 5 $a$'a and 5 $b$'s to the ends of these words does not change this bound. 
The words used in the proof of Theorem \ref{1point09} differ from these $u$ and $v$ only by switching two letters in one of the words, and yet this subtle change is enough to cause an exponential blow-up in size.

\bibliographystyle{eptcs}
\bibliography{shuffle}

\end{document}